\documentclass[superscriptaddress, showpacs, 12pt]{revtex4}
\usepackage{bm}
\usepackage{mathrsfs}
\usepackage{amssymb}
\usepackage{amsmath}
\usepackage{amsthm}

\begin{document}%

\newcommand{\ket}[1]{|#1\rangle}
\newcommand{\bra}[1]{\langle#1|}
\newcommand{\inner}[2]{\langle#1|#2\rangle}
\newcommand{\norm}[1]{\| #1\|}
\newcommand{\lr}\longrightarrow
\newcommand{\ra}\rightarrow
\newcommand{\tr}{{\rm Tr}}
\newcommand{\sgn}{{\rm sgn}}
\newcommand{\fsp}{{\rm span}}
\newcommand{\fsup}{{\rm supp}}
\newcommand{\fdg}{{\rm diag}}
\newcommand{\dt}{{\Delta t}}
\newcommand\reals{\mathbb R}
\newcommand\naturals{\mathbb N}
\newcommand{\e}{\varepsilon}

\newtheorem{thm}{Theorem}
\newtheorem{Prop}{Proposition}
\newtheorem{Coro}{Corollary}
\newtheorem{Lemma}{Lemma}
\newtheorem{Def}{Definition}
\newtheorem{Rem}{Remark}

\title{On the Efficiency of Quantum Algorithms for Hamiltonian Simulation}
\author{Anargyros Papageorgiou}
\email{ap@cs.columbia.edu}
\affiliation{Department of Computer
Science, Columbia University, New York, USA, 10027}
\author{Chi Zhang}
\email{czhang@cs.columbia.edu} \affiliation{Department of Computer
Science, Columbia University, New York, USA, 10027}
\date{\today}

\begin{abstract}

We study the efficiency of algorithms simulating a system
evolving with Hamiltonian $H=\sum_{j=1}^m H_j$. We consider high
order splitting methods that play a key role in quantum Hamiltonian
simulation. We obtain upper bounds on the number of exponentials required
to approximate $e^{-iHt}$ with error $\e$.
Moreover, we derive the order of the
splitting method that optimizes the cost of the resulting algorithm.
We show significant speedups relative to
previously known results.

\end{abstract}

\pacs{03.67.Ac, 03.67.Lx}

\maketitle

\section{Introduction}
While the computational cost of simulating many particle quantum
systems using classical computers grows exponentially with the
number of particles, quantum computers have the
potential to carry out the simulation efficiently
\cite{sim_1,sim_2,sim_3,Berry}. This property, pointed out by
Feynman, is one of the fundamental ideas of the field of quantum
computation. The simulation problem is also related to quantum walks
and adiabatic optimization
\cite{adia_1,adia_2,walk_1,walk_2,walk_3,walk_4}.

A variety of quantum algorithms have been proposed to predict and
simulate the behavior of different physical and chemical systems. Of
particular interest are {\it splitting methods} that simulate the
unitary evolution $e^{-iHt}$, where $H$ is the system Hamiltonian,
by a product of unitary operators of the form $e^{-iA_l t_l}$, for
some $t_l$, $l=1,\dots,N$, where $A_l\in\{H_1,\dots,H_m\}$,
$H=\sum_{j=1}^m H_j$ and assuming the Hamiltonians $H_j$ do not
commute. It is further assumed that the $H_j$ can be implemented
efficiently. Throughout this paper we assume that the $H_j$ are
either Hermitian matrices or bounded Hermitian operators so that
$\norm{H_j}< \infty$ for $j=1,\dots,m$, where $\norm{\cdot}$ is an
induced norm \footnote{An induced norm implies that
$\norm{H_jH_k}\leq \norm{H_j}\norm{H_k}$.}.

As Nielsen and Chuang \cite[p. 207]{NC} point out, the heart
of quantum simulation is
in the Lie-Trotter formula
$$\lim_{n\to\infty} \left( e^{-iH_1 t/n} e^{-iH_2 t/n}\right)^n
= e^{-i(H_1+H_2)t}.$$
From this one obtains the second order approximation
$$e^{-i(H_1+H_2)\dt} = e^{-iH_1 \dt} e^{-iH_2 \dt} + O(|\dt|^2).$$
A third order approximation is given by the Strang splitting
$$e^{-i(H_1+H_2)\dt} = e^{-iH_1 \dt/2} e^{-iH_2 \dt} e^{-iH_1 \dt/2}
+ O(|\dt|^3).$$
Suzuki \cite{Suzuki_1,Suzuki_2} uses recursive modifications of this
approximation to derive methods of order $2k+1$, for $k=1,2,\dots$.

A recent paper \cite{Berry} shows that Suzuki's high order splitting methods
can be used to derive bounds for the number $N$
of exponentials, assuming the $H_j$ are local Hamiltonians.
These bounds are expressed in terms of the evolution time  $t$, the
norm $\norm{H}$ of the Hamiltonian $H$, the order of the splitting
method $2k+1$, the number of Hamiltonians $m$, and the error $\e$ in
the approximation of $e^{-iHt}$. In this paper we will show how
these bounds can be significantly improved.

Consider the Hamiltonians indexed with respect to the magnitude of
their norms $\norm{H_1}\ge \norm{H_2}\ge \cdots \ge\norm{H_m}$. Then
the number of necessary exponentials $N$ generally depends on $H_1$,
but it must also depend explicitly on $H_2$ since only one
exponential should suffice for the simulation if $\norm{H_2}\to 0$.
This observation is particularly important for the simulation of
systems in physics and chemistry. To see this, suppose $m=2$ and
that $H_1$ is a discretization of the negative Laplacian $-\Delta$,
while $H_2$ is a discretization of a uniformly bounded potential.
Then $e^{-iH_1t_1}$ and $e^{-iH_2t_2}$ can be implemented
efficiently for any $t_1,t_2$, and $\norm{H_2}\ll \norm{H_1}$. We
will see that, not only in this case but in general, the number of
exponentials is proportional to both $\norm{H_1}$ and $\norm{H_2}$,
i.e., the Hamiltonian of the second largest norm plays an important
role.

Let $\e$ be sufficiently small. The previously known bound for the
number of exponentials, according to \cite{Berry}, is
\begin{equation}
N \le N_{\rm prev}:=
m5^{2k}(m\norm{H}t)^{1+\frac{1}{2k}}{\e}^{-1/(2k)}.
\end{equation}
This bound does not properly reflect the dependence on $H_2$.

Performing a more detailed analysis of the approximation error by
high order splitting formulas, it is possible to improve the bounds
for $N$ substantially. The new estimates lead to {\em optimal}
splitting methods of significantly lower order which greatly reduces
the cost of the algorithms.

We now summarize our results. Recall that the $H_j$ can be
implemented efficiently  but do not commute and $\norm{H_1}\ge
\norm{H_2}\ge\cdots \norm{H_m}$. We show the following:
\begin{enumerate}

\item A new bound for the number of exponentials $N$, given by
$$N\le N_{\rm new}:= 2 (2m-1)\; 5^{k-1}
 \norm{H_1}t \left( \frac { 4emt \norm{H_2} }{\e }
\right)^{1/(2k)} \frac {4me}3 \left( \frac 53 \right)^{k-1}.$$

\item A speedup factor of
\begin{equation*}
\frac{N_{\rm new}}{N_{\rm prev}} \leq \frac {2}{3^k}
\left(\frac{4e\norm{H_2}}{\norm{H_1}}\right)^{1/2k}.
\end{equation*}

\item We show that the {\em optimal} $k^*_{\rm new}$
that minimizes $N_{\rm new}$ is
$$k^*_{\rm new} := {\rm round}\left( 
\sqrt{\frac 12  \log_{25/3} \frac {4emt\norm{H_2}}{\e}} \, \right). $$

On the other hand, from \cite{Berry} the bound for $N_{\rm prev}$ is
minimized for
\begin{equation*}
k^*_{\rm prev} =  {\rm round}\left(
\frac{1}{2}\sqrt{\log_5\frac{m\norm{H}t}{\e}+1}\,\right).
\end{equation*}

\item For $k_{\rm new}^*$ the value of $N_{\rm new}$ satisfies
$$
N^*_{\rm new} \leq \frac {8}{3}\, (2m-1)\, met \, \norm{H_1} \; e^{2
\sqrt{\tfrac 12 \ln \tfrac {25}3 \ln \tfrac {4emt\norm{H_2}}{\e}}}.
$$ For $k_{\rm prev}^*$ the value of $N_{\rm prev}$ is
\begin{equation*}
N^*_{\rm prev} = 2m^2\norm{H}t\cdot e^{2\sqrt{\ln5
\ln(m\norm{H}t/\epsilon)}}.
\end{equation*}
Hence
$$\frac {N^*_{\rm new}}{N^*_{\rm prev}} \le
 \frac {8\,e}{3}\,
e^{2 \left( \sqrt{\tfrac 12 \ln \tfrac {25}3 \ln \tfrac
{4emt\norm{H_2}}{\e}} -  \sqrt{ \ln 5 \ln
\tfrac{m\norm{H_1}t}{\e}}\, \right) }.
$$
\end{enumerate}

\section{Splitting methods for simulating the sum of two Hamiltonians}

We begin this section by discussing the simulation of
$$e^{-i(H_1+H_2)t},$$
where $H_1,H_2$ are given Hamiltonians. Restricting the analysis to
$m=2$ will allow us to illustrate the main idea in our approach
while avoiding the rather complicated notation needed in the general
case, for $m\ge 2$. The simulation of the Schr\"odinger equation of
a $p$-particle system, where $H_1$ is obtained from the Laplacian
operator and $H_2$ is the potential, requires one to consider an
evolution operator that has the form above; see \cite{sim_3}.

In the next section we deal with the more general simulation problem
involving a sum of $m$ Hamiltonians, $H_1,\dots, H_m$, as Berry et
al. \cite{Berry} did, and we will show how to improve their
complexity results.

Suzuki proposed methods for decomposing exponential operators in a
number of papers \cite{Suzuki_1, Suzuki_2}.
For sufficiently small $\dt$,
starting from the
formula
\begin{equation*}
S_2(H_1,H_2,\dt) = e^{-iH_1\dt/2}e^{-iH_2\dt}e^{-iH_1\dt/2},
\end{equation*}
and proceeding recursively, Suzuki defines
\begin{equation*}
S_{2k}(H_1,H_2,\dt) =
[S_{2k-2}(H_1,H_2,p_k\dt)]^2S_{2k-2}(H_1,H_2,(1-4p_k)\dt)[S_{2k-2}(H_1,H_2,p_k\dt)]^2,
\end{equation*}
for $k=2,3,\cdots$, where $p_k = (4-4^{1/(2k-1)})^{-1}$, and then
proves that
\begin{equation}\label{suzuki}
\bigl\| e^{-i(H_1+H_2) \dt}-S_{2k}(H_1,H_2,\dt) \bigr\| =
O(|\dt|^{2k+1}).
\end{equation}
Suzuki was particularly interested in the order of his method, which
is $2k+1$, and did not address the size of the implied asymptotic
factors in the big-$O$ notation. However, these factors depend on the norms
of $H_1$ and $H_2$ and can be very large, when $H_1$ and $H_2$ do
not commute. For instance, when $H_1$ is obtained from the
discretization of the Laplacian operator with mesh size $h$,
$\norm{H_1}$ grows as $h^{-2}$. Since $h=\e$, we get $\norm{H_1} =
O(\frac{1}{\e^2})$. Hence, for fine discretizations $\norm{H_1}$ is
huge, and severely affects the error bound above.

Suppose $\norm{H_1}\ge \norm{H_2}$. Since
\begin{equation*}
e^{-i(H_1+H_2)t} = e^{-i(\mathcal{H}_1+\mathcal{H}_2)\norm{H_1}t},
\end{equation*}
where $\mathcal{H}_j = H_j / \norm{H_1}$, for $j=1,2$, we can
consider the simulation problem for $\mathcal{H}_1+\mathcal{H}_2$
with an evolution time $\tau =\norm{H_1}t$.

Unwinding the recurrence in Suzuki's construction yields
\begin{equation}
\label {eq:largeformula}
S_{2k}(\mathcal{H}_1,\mathcal{H}_2,\dt) =
\prod_{\ell=1}^{K} S_2(\mathcal{H}_1,\mathcal{H}_2,z_{\ell}\dt)
= \prod_{\ell=1}^{K}
\left[ e^{-i\mathcal{H}_1z_\ell \dt/2}e^{-i\mathcal{H}_2z_\ell \dt}e^{-i\mathcal{H}_1 z_\ell \dt/2} \right],
\end{equation}
where $K=5^{k-1}$ and each $z_\ell$ is defined
according to the recursive scheme, $\ell=1,\dots,K$.
In particular, $z_1 = z_K = \prod_{r=2}^k p_r$, and for the
intermediate values of $\ell$ the $z_\ell$ is a product of $k-1$ factors
and has the form $z_\ell = \prod_{r\in I_0} p_r \prod_{r\in I_1} (1-4p_r)$,
where the products are over the index sets $I_0, I_1$ defined
by traversing the corresponding to $\ell$ path of the
recursion tree.

Let $q_r = \max \{ p_r, 4p_r - 1 \}$, $r\ge 2$. Then
$\{ q_r \}$ is a decreasing sequence of positive numbers and from
\cite[p. 18]{wiebe} we have that
\begin{equation*}
%\label {eq:qbound0}
\frac 3{3^{k}} \le \prod_{r=2}^k q_r \le \frac {4k}{3^k}.
\end{equation*}
Thus
\begin{equation} \label{eq:qbound0}
|z_\ell| \le \frac {4k}{3^k}\quad {\rm\ for\ all\ } \ell=1,\dots,K.
\end{equation}

Equation (\ref{eq:largeformula}) can be expressed in the more compact form
which we use to simplify the notation. Namely,
\begin{equation}\label{eq:compactformula}
S_{2k}(\mathcal{H}_1,\mathcal{H}_2,\dt) =
e^{-i\mathcal{H}_1s_0\dt}e^{-i\mathcal{H}_2z_1\dt}e^{-i\mathcal{H}_1s_1\dt}\cdots
e^{-i\mathcal{H}_2z_K\dt}e^{-i\mathcal{H}_1s_K\dt},
\end{equation}
where $s_0=z_1/2$, $s_j=(z_j+z_{j+1})/2$, $j=1,\dots,K-1$, and
$s_{K}=z_{K}/2$.
Observe that
$\sum_{j=0}^K s_j = 1$, $\sum_{j=1}^K z_j=1$.

We need to bound
$\sigma_k = \sum_{j=0}^K |s_j|+\sum_{j=1}^K |z_j|$
from above. From (\ref{eq:qbound0}) we have
$$
\sum_{j=1}^K |z_j| \le \frac {4k5^{k-1}}{3^k},
$$
and also
$$
\sum_{j=0}^K |s_j| \le \frac {4k5^{k-1}}{3^k}.
$$
Thus
\begin{equation}\label{eq:qbound}
\sigma_k \le \frac 83 k \left(\frac 53\right)^{k-1}=:c_k \quad {\rm \ for\ } k\ge 1.
\end{equation}
(The above trivially holds for $k=1$.)

Expanding each exponential in (\ref{eq:compactformula}) we obtain
\begin{equation}\label{exp}
\begin{split}
S_{2k}(\mathcal{H}_1,\mathcal{H}_2,\dt) =&(I+\mathcal{H}_1s_0(-i\Delta t)+\frac{1}{2}\mathcal{H}_1^2s_0^2(-i\Delta t)^2+\cdots+\frac{1}{k!}\mathcal{H}_1^ks_0^k(-i\Delta t)^k+\cdots)\\
&(I+\mathcal{H}_2z_1(-i\Delta t)+\frac{1}{2}\mathcal{H}_2^2z_1^2(-i\Delta t)^2+\cdots+\frac{1}{k!}\mathcal{H}_2^kz_1^k(-i\Delta t)^k+\cdots)\\
&(I+\mathcal{H}_1s_1(-i\Delta t)+\frac{1}{2}\mathcal{H}_1^2s_1^2(-i\Delta t)^2+\cdots+\frac{1}{k!}\mathcal{H}_1^ks_1^k(-i\Delta t)^k+\cdots)\\
&\cdots\\
&(I+\mathcal{H}_2z_K(-i\Delta t)+\frac{1}{2}\mathcal{H}_2^2z_K^2(-i\Delta t)^2+\cdots+\frac{1}{k!}\mathcal{H}_2^kz_K^k(-i\Delta t)^k+\cdots)\\
&(I+\mathcal{H}_1s_K(-i\Delta
t)+\frac{1}{2}\mathcal{H}_1^2s_K^2(-i\Delta
t)^2+\cdots+\frac{1}{k!}\mathcal{H}_1^ks_K^k(-i\Delta t)^k+\cdots).
\end{split}
\end{equation}
After carrying out the multiplications we see that $S_{2k}$ is a
sum of terms that has the form
\begin{equation}\label{exp2}
\frac{s_0^{\alpha_0}s_1^{\alpha_1}\cdots
s_K^{\alpha_K}z_1^{\beta_1}\cdots
z_K^{\beta_K}}{\alpha_0!\alpha_1!\cdots\alpha_K!\beta_1!\cdots\beta_K!}
\mathcal{H}_1^{\alpha_0}\mathcal{H}_2^{\beta_1}\mathcal{H}_1^{\alpha_1}\cdots
\mathcal{H}_2^{\beta_K}\mathcal{H}_1^{\alpha_K} (-i\Delta
t)^{\sum_{i=0}^K\alpha_i+\sum_{j=1}^K\beta_j},
\end{equation}
where the $\alpha_0,\alpha_1,\cdots,\alpha_K$ and the
$\beta_1,\cdots,\beta_K$ are obtained by multiplying the
denominators in the expansion of the exponentials.

The terms that do not contain $\mathcal{H}_2$ are those for which
$\beta_1=\beta_2=\cdots=\beta_K=0$, and their sum is
\begin{equation}
\begin{split}
&\sum_{\alpha_0,\alpha_1,\cdots,\alpha_K}
\frac{s_0^{\alpha_0}s_1^{\alpha_1}\cdots
s_K^{\alpha_K}}{\alpha_0!\alpha_1!\cdots\alpha_K!}\mathcal{H}_1^{\sum_{j=0}^K\alpha_j}(-i\Delta t)^{\sum_{j=0}^K\alpha_j}\\
&=\sum_{\alpha_0}
\frac{1}{\alpha_0!}\mathcal{H}_1^{\alpha_0}(-is_0\Delta
t)^{\alpha_0}\cdot \sum_{\alpha_1}
\frac{1}{\alpha_1!}\mathcal{H}_1^{\alpha_1}(-is_1\Delta
t)^{\alpha_1}\cdot \cdots \cdot\sum_{\alpha_K}
\frac{1}{\alpha_K!}\mathcal{H}_1^{\alpha_K}(-is_K\Delta t)^{\alpha_K}\\
&= \prod_{j=0}^K e^{-i\mathcal{H}_1s_j\Delta t} =
\exp(-i\sum_{j=0}^K \mathcal{H}_1s_j\Delta t) =
\exp(-i\mathcal{H}_1\Delta t).
\end{split}
\end{equation}
On the other hand, consider
\begin{equation}
e^{-i(\mathcal{H}_1+\mathcal{H}_2)\dt} = I +
(-i(\mathcal{H}_1+\mathcal{H}_2)\dt) + \cdots + \frac
1{k!}(-i(\mathcal{H}_1+\mathcal{H}_2)\dt)^k +\cdots.
\end{equation}
The terms that do not contain $\mathcal{H}_2$ sum to
\begin{equation}
\sum_{k=0}^{\infty} \frac{1}{k!}\mathcal{H}_1^k (-i\dt)^k =
e^{-i\mathcal{H}_1\Delta t}.
\end{equation}

Let us now consider the bound in (\ref{suzuki}). Clearly the terms
that do not contain $\mathcal{H}_2$ cancel out. Therefore, the error
is proportional to $\norm{\mathcal{H}_2}|\dt|^{2k+1}$, i.e. it
depends on the ratio $\norm{H_2}/\norm{H_1}$ of the norms of the
original Hamiltonians. This fact will be used to
improve the error and complexity results of Berry et al. \cite{Berry}

\begin{Lemma} For $k\in\naturals$, $c_k|\dt|\le k+1$ (see, Eq. \ref{eq:qbound})
and $\norm{\mathcal{H}_2}\le \norm{\mathcal{H}_1}=1$ we have
\begin{equation}
\norm{\exp(-i(\mathcal{H}_1+\mathcal{H}_2)\dt) -
S_{2k}(\mathcal{H}_1,\mathcal{H}_2,\dt)} \leq
 \frac{4\norm{\mathcal{H}_2}}{(2k+1)!}(c_k|\dt|)^{2k+1}.
\end{equation}
\end{Lemma}

\begin{proof}
For notational convenience we use $S_{2k}(\dt)$ to denote
$S_{2k}(\mathcal{H}_1,\mathcal{H}_2,\dt)$. Consider
\begin{equation}
\exp(-i(\mathcal{H}_1+\mathcal{H}_2)\dt)-S_{2k}(\dt) =
\sum_{l=2k+1}^{\infty} \bigl[ R_{l}(\dt)-T_{l}(\dt)\bigr],
\end{equation}
where $R_{l}(\dt)$ is the sum of all terms in
$\exp(-i(\mathcal{H}_1+\mathcal{H}_2)\dt)$ corresponding to
$\dt^{l}$ and $T_{l}(\dt)$ is the sum of all terms in $S_{2k}(\dt)$
corresponding to $\dt^l$. Moreover, we know that the terms with only
$\mathcal{H}_1$ cancel out. Hence, we can ignore the terms in
$T_{l}(\dt)$ and $R_{l}(\dt)$ that contain only $\mathcal{H}_1$ (and
not $\mathcal{H}_2$) as a factor. It follows that
\begin{equation}
R_{l}(\dt) = \frac{1}{l!}
(\mathcal{H}_1+\mathcal{H}_2)^{l}(-i\dt)^{l}
-\frac{1}{l!}\mathcal{H}_1^{l}(-i\dt)^{l}.
\end{equation}
Then
\begin{equation}\label{t}
\norm{R_{l}(\dt)} \leq \frac{1}{l!} 2^{l} \norm{\mathcal{H}_2}
|\dt|^{l},
\end{equation}
since there are $2^{l}-1$ terms, and they are bounded by
$\frac{1}{l!}\norm{\mathcal{H}_2}|\dt|^l$.

Now consider the terms in $T_{l}(\dt)$. From (\ref{exp},\ref{exp2})
\begin{equation}
T_{l}(\dt) = \sum_{\sum_{i=0}^K\alpha_i+\sum_{i=1}^K \beta_i = l,\;
\sum_{i=1}^K \beta_i\ne 0}\frac{s_0^{\alpha_0}s_1^{\alpha_1}\cdots
   %remember edit
s_K^{\alpha_K}z_1^{\beta_1}\cdots
z_K^{\beta_K}}{\alpha_0!\alpha_1!\cdots\alpha_K!\beta_1!\cdots\beta_K!}
\mathcal{H}_1^{\alpha_0}\mathcal{H}_2^{\beta_1}\mathcal{H}_1^{\alpha_1}\cdots
\mathcal{H}_2^{\beta_K}\mathcal{H}_1^{\alpha_K} (-i\dt)^{l},
\end{equation}
where the condition $\sum_{i=1}^K\beta_i \neq 0$ hold because there
are no terms containing $\mathcal{H}_1$ alone. Since the norm of
$\mathcal{H}_1^{\alpha_0}\mathcal{H}_2^{\beta_1}\mathcal{H}_1^{\alpha_1}\cdots
\mathcal{H}_2^{\beta_K}\mathcal{H}_1^{\alpha_K}$ is at most
$\norm{\mathcal{H}_2}$, we have
\begin{equation}
\norm{T_{l}(\dt)} \leq \sum_{\sum_{i=0}^K\alpha_i+\sum_{i=1}^K
\beta_i = l}\frac{|s_0^{\alpha_0}s_1^{\alpha_1}\cdots
s_K^{\alpha_K}z_1^{\beta_1}\cdots
z_K^{\beta_K}|}{\alpha_0!\alpha_1!\cdots\alpha_K!\beta_1!\cdots\beta_K!}
\norm{\mathcal{H}_2}|\dt|^l.
\end{equation}
Note that we relaxed the condition $\sum_{i=1}^K \beta_i \neq 0$
since it does not affect the inequality.

To calculate the sum
$\sum\frac{|s_0^{\alpha_0}s_1^{\alpha_1}\cdots
s_K^{\alpha_K}z_1^{\beta_1}\cdots
z_K^{\beta_K}|}{\alpha_0!\alpha_1!\cdots\alpha_K!\beta_1!\cdots\beta_K!}$,
where $\sum_{i=0}^K\alpha_i+\sum_{i=1}^K \beta_i = l$, we first
consider the following equation
\begin{equation}
\begin{split}
&\exp(|s_0\dt|)\exp(|z_1\dt|)\exp(|s_1\dt|)\cdots\exp(|z_K\dt|)\exp(|s_K\dt|)\\
&= \left(\sum_{\alpha_0=0}^{\infty}
\frac{1}{\alpha_0!}|s_0\dt|^{\alpha_0}\right)
\cdot\left(\sum_{\beta_1=0}^{\infty}
\frac{1}{\beta_1!}|z_1\dt|^{\beta_0}\right)
\cdot\left(\sum_{\alpha_1=0}^{\infty}
\frac{1}{\alpha_1!}|s_1\dt|^{\alpha_0}\right) \cdots  \\
&\quad \quad \quad\quad\cdots \cdot \left(\sum_{\beta_K=0}^{\infty}
\frac{1}{\beta_K!}|z_K\dt|^{\beta_K}\right)
\cdot\left(\sum_{\alpha_K=0}^{\infty}
\frac{1}{\alpha_K!}|s_K\dt|^{\alpha_K}\right) \\
&= \sum_{p=0}^{\infty} \sum_{\sum\alpha_j+\sum\beta_j
=p}\frac{|s_0^{\alpha_0}s_1^{\alpha_1}\cdots
s_K^{\alpha_K}z_1^{\beta_1}\cdots
z_K^{\beta_K}|}{\alpha_0!\alpha_1!\cdots\alpha_K!\beta_1!\cdots\beta_K!}
|\dt|^{p}.
\end{split}
\end{equation}
Hence $\sum_{\sum\alpha_j+\sum\beta_j
=l}\frac{|s_0^{\alpha_0}s_1^{\alpha_1}\cdots
s_K^{\alpha_K}z_1^{\beta_1}\cdots
z_K^{\beta_K}|}{\alpha_0!\alpha_1!\cdots\alpha_K!\beta_1!\cdots\beta_K!}$
is the coefficient of $|\dt|^{l}$ in the equation above. Similarly,
\begin{equation}
\begin{split}
&\exp(|s_0\dt|)\exp(|z_1\dt|)\exp(|s_1\dt|)\cdots\exp(|z_K\dt|)\exp(|s_K\dt|)\\
&=\exp((\sum_{i=0}^K |s_i| + \sum_{i=1}^K |z_i|)|\dt|) =
\exp(\sigma_k|\dt|)\\
&= \sum_{p=0}^{\infty} \frac{1}{p!} \sigma_k^p |\dt|^{p},
\end{split}
\end{equation}
Recall that the bound for $\sigma_k$ given in Eq. (\ref{eq:qbound}).
Thus the coefficient of $|\dt|^{l}$ is bounded from above by
$\frac{1}{l!} c_k^l$. Therefore,
we have
\begin{equation}\label{re}
\norm{T_{l}(\dt)} \leq \frac{c_k^l}{l!}\norm{\mathcal{H}_2}|\dt|^{l}.
\end{equation}

We combine Eq. (\ref{t}), (\ref{re}), to obtain
\begin{equation}
\begin{split}
\norm{\exp((\mathcal{H}_1+\mathcal{H}_2)\dt)-S_{2k}(\dt)} &\leq
\sum_{l=2k+1}^{\infty}\norm{R_{l}(\dt)-T_{l}(\dt)} \\ &\leq
\sum_{l=2k+1}^{\infty}\norm{R_{l}(\dt)}+\norm{T_{l}(\dt)}\\
&\leq 2 \sum_{l=2k+1}^{\infty}
\frac{c_k^l}{l!}\norm{\mathcal{H}_2}|\dt|^{l}\\
&\leq \frac{2}{(2k+1)!}\norm{\mathcal{H}_2}|c_k\dt|^{2k+1}
\left( 1 -\frac {c_k|\dt|}{2k+2}\right)^{-1} \\
&\le
\frac{4}{(2k+1)!}\norm{\mathcal{H}_2} |c_k\dt|^{2k+1},
\end{split}
\end{equation}
where the last two inequalities follow from
the assumption $c_k |\dt|\le k+1$.
and an estimate of the tail of
the Poisson distribution; see, e.g., \cite[Thm 1]{possion}.
\end{proof}

\begin{thm} Let $1\ge \e>0$
be such that $8et\norm{H_2}\ge \e$.
The number $N$ of exponentials for the simulation of
$e^{-i(H_1+H_2)t}$ with accuracy $\e$ is bounded as follows
$$
N\le 3\; 5^{k-1}
\left\lceil
\norm{H_1}t \left(\frac {8et\norm{H_2}}{ \e } \right)^{1/(2k)}
\frac {8e}3 \left(\frac 53\right)^{k-1}\right\rceil,
$$
for any $k\in \naturals$, where $\norm{H_2}\le\norm{H_1}$.
\end{thm}

\begin{proof} Let $M= |{\dt}|^{-1}$. Then using Lemma 1 and
$\mathcal{H}_j = H_j / \norm{H_1}$, $j=1,2$, we obtain
\begin{equation*}
\begin{split}
\big\| e^{-i(H_1+H_2)t} -
S_{2k}^{M\norm{H_1}t}(\mathcal{H}_1,\mathcal{H}_2, 1/M)\big\| &\le
M\norm{H_1} t \frac{4}{(2k+1)!}\norm{\mathcal{H}_2}
\left(\frac {c_k}{M}\right)^{2k+1} \\
&=  4t \norm{H_2} \,\frac{c_k^{2k+1}}{(2k+1)!} \frac 1{M^{2k}}.
\end{split}
\end{equation*}
Recall that $c_k$ is defined in (\ref{eq:qbound}) and is used in Lemma 1.
For accuracy $\e$ we obtain
$$M\ge \left(\frac {4t\norm{H_2} c_k^{2k+1}}{\e (2k+1)!} \right)^{1/(2k)}.$$
We use Stirling's formula \cite[p. 257]{Stegun} for the factorial function
$$
(2k+1)!=\sqrt{2\pi}(2k+1)^{(2k+1)+1/2}e^{-(2k+1)+\theta/(12(2k+1))},
\quad 0<\theta<1,$$
which yields
\begin{equation}\label{eq:factbound}
[(2k+1)!]^{-1/(2k)} \le e^{1+1/(2k)} /(2k+1).
\end{equation}
It is easy to check that
$$c_k^{1/(2k)} \le 2^{1+1/(2k)}.$$
Thus it suffices to
take
$$M \ge \left(\frac {8et\norm{H_2}}{\e } \right)^{1/(2k)}
\frac {2\,e\; c_k}{2k+1}.$$

So we define $M$ to be lower bound of the expression above, i.e.,
$$M := \left(\frac {8et\norm{H_2}}{\e } \right)^{1/(2k)}
\frac {2e\;c_k}{2k+1}.$$
It is easy to check that
$$\frac {2e}{2k+1} (k+1)\ge e,$$
which along with the condition $8et\norm{H_2}\ge \e$ yields $M (k+1) \ge c_k$.
This shows the assumptions of Lemma~1 are satisfied with this value of $M$.

From the recurrence relation the number of required
exponentials to implement
$S_{2k}$ in one subinterval is no more than $3\cdot5^{k-1}$.
We need to consider two cases
concerning $M\norm{H_1}t$.
If $M\norm{H_1}t\ge 1$,
then the number of subintervals is
$\lceil M\norm{H_1}t\rceil$, i.e.,
we partition the entire time
interval into an integer number of subintervals, each of length
at most $M^{-1}$.
The total number of required exponentials is bounded by
$3\cdot5^{k-1}\lceil M\norm{H_1}t\rceil$.
Substituting the values of $M$ and $c_k$ we obtain
the bound for $N$. In particular,
\begin{equation}\label{eq:N1}
N\le 3\cdot 5^{k-1}
\left\lceil \norm{H_1}t \left(\frac {8et\norm{H_2}}{ \e } \right)^{1/(2k)}
\frac {8e}3 \left(\frac 53\right)^{k-1}\right\rceil.
\end{equation}

If $M\norm{H_1}t< 1$, then Lemma 1 can be used with $\dt=\norm{H_1}t$,
since $\norm{H_1}t\le M^{-1}$ and we have already seen that $M$ is such
that the assumptions of Lemma 1 are satisfied. Thus
\begin{eqnarray*}
\big\| e^{-i(H_1+H_2)t} -
S_{2k}(\mathcal{H}_1,\mathcal{H}_2, \norm{H_1}t)\big\| &\le&
\frac{4}{(2k+1)!}\norm{\mathcal{H}_2}
\left(c_k \norm{H_1}t \right)^{2k+1} \\
=  4t \norm{H_2} \,\frac{c_k^{2k+1}}{(2k+1)!} (\norm{H_1}t)^{2k}
&\le&  4t \norm{H_2} \,\frac{c_k^{2k+1}}{(2k+1)!} (M)^{-2k} \le \e,
\end{eqnarray*}
where the last inequality holds by definition of $M$.
In this case the total number of exponentials is simply
\begin{equation}\label{eq:N2}
N\le  3\cdot 5^{k-1}.
\end{equation}

Combining (\ref{eq:N1}) and (\ref{eq:N2}) we obtain
$$
N\le 3\cdot 5^{k-1}
\left\lceil
\norm{H_1}t \left(\frac {8et\norm{H_2}}{ \e } \right)^{1/(2k)}
\frac {8e}3 \left(\frac 53\right)^{k-1}\right\rceil.
$$
This completes the proof.
\end{proof}

\begin{Rem}
Lemma 1 and Theorem 1 indicate that when $\norm{H_2}t\ll \e$ then the
number of exponentials $N$ can be further improved. In this case it
can be shown that high order splitting methods may lose their
advantage. We do not pursue this direction in this paper since we
assume that the $H_j$, $j=1,\dots, m$, are fixed and study $N$ as $\e\to 0$.
\end{Rem}

\section{Splitting methods for simulating the sum of many Hamiltonians}

In this section we deal with the simulation of
$$e^{-i\sum_{j=1}^m H_j t},$$
where $H_j$, $j=1,\dots,m$, are given non-commuting
Hamiltonians. The analysis and the conclusions are similar to those
of the previous section where $m=2$, but the proofs are much more
complicated and certainly tedious. This is the problem that Berry et
al. \cite{Berry} considered.

We use Suzuki's recursive construction once more \cite{Suzuki_2}. In
particular, for
$$
S_2(H_1,\dots,H_m,\dt) = \prod_{j=1}^m e^{-iH_j\dt/2}\prod_{j=m}^1 e^{-iH_j\dt/2},
$$
and
\begin{equation*}
S_{2k}(H_1,\dots,H_m,\dt) =
[S_{2k-2}(p_k\dt)]^2S_{2k-2}((1-4p_k)\dt)[S_{2k-2}(p_k\dt)]^2,
\quad k=2,3,\dots,
\end{equation*}
where for notational convenience we have used $S_{2k-2}(\dt)$ to
denote $S_{2k-2}(H_1,\cdots,H_m,\dt)$, and $p_k =
(4-4^{1/(2k-1)})^{-1}$, we have that
\begin{equation}\label{suzuki2}
\bigl\| e^{-i\sum_{j=1}^{m}H_j \dt}-S_{2k}(H_1,\dots,H_m,\dt) \bigr\| =
O(|\dt|^{2k+1}).
\end{equation}

Assuming again that $\norm{H_1}\ge \norm{H_2}\ge \cdots\ge \norm{H_m}$
we normalize the Hamiltonians by setting $\mathcal{H}_j = H_j / \norm{H_1}$,
$j=1,\dots,m$, and consider the equivalent simulation problem
$$e^{-i\sum_{j=1}^m \mathcal{H}_j \tau},$$
where $\tau= \norm{H_1} t$. Proceeding in a way similar to that for
$m=2$ of the previous section we derive the following lemma, whose
proof can be found in the Appendix.

\begin{Lemma} For $k\in\naturals$, $d_k|\dt | \le k+1$,
$d_k= m(4/3)k(5/3)^{k-1}$ and
$\norm{\mathcal{H}_m} \le \cdots\le
\norm{\mathcal{H}_2} \le \norm{\mathcal{H}_1}=1$
we have
\begin{equation}
\norm{\exp(-i\sum_{j=1}^m \mathcal{H}_j\dt)
- S_{2k}(\mathcal{H}_1,\dots,\mathcal{H}_m,\dt)} \leq
\frac{4\norm{\mathcal{H}_2}}{(2k+1)!}(d_k|\dt|)^{2k+1}.
\end{equation}
\end{Lemma}

From Lemma 2, we have the following theorem.

\begin{thm} Let $1\ge \e>0$ be such that
$4met\norm{H_2}\ge\e$.
The number $N$ of exponentials for the simulation of
$e^{-i(H_1+\cdots+H_m)t}$ with accuracy $\e$ is bounded by
$$N \le (2m-1)\; 5^{k-1}
\left\lceil \norm{H_1}t
\left( \frac { 4emt \norm{H_2} }{\e } \right)^{1/(2k)}
\frac {4me}3 \left( \frac 53 \right)^{k-1} \right\rceil,
$$
for any $k\in \naturals$, where
$\norm{H_m}\le\cdots\le\norm{H_2}\le\norm{H_1}$.
\end{thm}

\begin{proof}
The proof is similar to that of Theorem 1.
Let $M= |\dt|^{-1}$. Then using Lemma 2 and
$\mathcal{H}_j = H_j / \norm{H_1}$, $j=1,\dots,m$, we obtain
\begin{equation*}
\begin{split}
\big\| e^{-i(H_1+\cdots +H_m)t} -
S_{2k}^{M\norm{H_1}t}(\mathcal{H}_1,\dots,\mathcal{H}_m, 1/M)\big\| &\le
M\norm{H_1} t \frac{4}{(2k+1)!}\norm{\mathcal{H}_2}
\left(\frac {d_k}{M}\right)^{2k+1} \\
&= 4 t \norm{H_2}  \,\frac{d_k^{2k+1}}{(2k+1)!}
\frac 1{M^{2k}}.
\end{split}
\end{equation*}
Recall that $d_k$ is defined in Lemma 2.
For accuracy $\e$ we obtain
$$M\ge \left(
\frac {4t\norm{H_2} d_k^{2k+1}}{\e (2k+1)!} \right)^{1/(2k)}.$$
We use the estimate (\ref{eq:factbound}). It is easy to check that
$$d_k^{1/(2k)} \le 2 m^{1/(2k)}.$$
Thus it suffices to take
$$M \ge \left( \frac {4emt\norm{H_2}}{\e } \right)^{1/(2k)}
\frac {2e\; d_k}{2k+1}.$$

So we define $M$ to be the lower bound of the expression above, i.e.,
$$
M:= \left( \frac { 4emt \norm{H_2} }{\e } \right)^{1/(2k)}
\frac {2e\; d_k}{2k+1}.$$
As in the proof of Theorem 1, it is straightforward to verify that
$M (k+1) \ge d_k$. Therefore,
the assumptions of Lemma 2 are satisfied for this value of $M$.

From the recurrence relation, we see that the number of required
exponentials to implement $S_{2k}$ in one subinterval
is no more than $(2m-1)\cdot5^{k-1}$.
Again we distinguish two cases for $M\norm{H_1}t$. We deal with
the case $M\norm{H_1}t<1$ in the same way we did in the proof of Theorem 1,
to conclude
$$N \le (2m-1)\cdot5^{k-1}.$$

If $M\norm{H_1}t\ge 1$, then the total number of required exponentials
is
$$N \le (2m-1)\cdot5^{k-1}\lceil M\norm{H_1}t\rceil.$$
Substituting the values of $M$ and $d_k$ we obtain
$$N \le (2m-1)\cdot5^{k-1}
\left\lceil
\norm{H_1} t
\left( \frac { 4emt \norm{H_2} }{\e } \right)^{1/(2k)}
\frac {4me}3 \left( \frac 53 \right)^{k-1} \right\rceil.
$$
This completes the proof.
\end{proof}

The reader may wish to recall Remark 1 that applies in the
case too.

\begin{Coro} If in addition to the  assumptions of Theorem 2
either of the following two conditions holds:
\begin{itemize}
\item $4met\norm{H_1} \ge 3$
\item $\e$ is sufficiently small
such that
$$\left( \ln \frac {4met\norm{H_1}}{5} \right)^2 -
2 \ln\frac 53 \ln \frac {4met\norm{H_2}}{\e} <0$$
\end{itemize}
then
the number of exponentials, $N$, for the simulation of
$e^{-i(H_1+\cdots+H_m)t}$ with accuracy $\e$ is bounded by
$$N \le 2\; (2m-1)\; 5^{k-1}
 \norm{H_1}t
\left( \frac { 4emt \norm{H_2} }{\e } \right)^{1/(2k)}
\frac {4me}3 \left( \frac 53 \right)^{k-1},
$$
for any $k\in \naturals$.
\end{Coro}

\begin{proof}
From the assumption of Theorem 2 we have $4emt\norm{H_2} / \e \ge 1$.
Consider the argument of the ceiling function in the bound of Theorem 2.
It is greater than or equal to $1$, if $4met\norm{H_1} \ge 3$.
Otherwise,
we take its logarithm and multiply the resulting expression by $k$.
This gives the quadratic polynomial
$$2k^2 \ln\frac 53 + 2k \ln \frac {4met\norm{H_1}}{5} +
\ln \frac {4met\norm{H_2}}{\e}.$$
When $\e$ is sufficiently small and the discriminant is negative, i.e., when
$$\left( \ln \frac {4met\norm{H_1}}{5} \right)^2 -
2 \ln\frac 53 \ln \frac {4met\norm{H_2}}{\e} <0,$$
the polynomial is positive for all $k$. Hence, that argument of
the ceiling function in the bound of Theorem 2
is greater than $1$, for all $k\ge 1$.

In either case,
we use $\lceil x\rceil \le 2x$, for $x\ge 1$, to estimate $N$ from above.
\end{proof}

\section{Speedup}

Let us now deal with the cost for simulating the evolution
$e^{-i(\sum_{j=1}^m H_j)t}$.
Berry et al. \cite{Berry} show upper and lower bounds for the number of
required exponentials. We concentrate on upper bounds and improve the
estimates of \cite{Berry}.

We are interested in the number of
exponentials required by the splitting formula that approximates the
evolution with accuracy $\e$.
Recall that 
\begin{equation*} \label{eq:Nnew}
N_{\rm new}:=
 2\; (2m-1)\; 5^{k-1}
 \norm{H_1}t
\left( \frac { 4emt \norm{H_2} }{\e } \right)^{1/(2k)}
\frac {4me}3 \left( \frac 53 \right)^{k-1}
\end{equation*}
exponentials suffice for
error $\e$. 
The above estimate holds for $\e$ sufficiently small as Theorem~2 and 
Corollary~1 indicate.
The corresponding previously known estimate \cite{Berry} is
$$N_{\rm prev}=
m\; 5^{2k}\; (m\norm{H_1}t)^{1+\frac{1}{2k}}\left(\frac{1}{\epsilon}\right)^{\frac{1}{2k}},$$
where $H=\sum_{j=1}^l H_j$.

The ratio of the two estimates is
\begin{equation}
\frac{N_{\rm new}}{N_{\rm prev}} \leq
\frac {2}{3^k}
\left(\frac{4e\norm{H_2}}{\norm{H_1}}\right)^{1/2k}.
\end{equation}
So for large $k$ we have an improvement in the
estimate of the cost of the algorithm. On the other hand, if
$\norm{H_2}\ll\norm{H_1}$ we have an improvement in the estimate of
the cost the algorithm not just for large $k$ but for all $k$. This
is particularly significant when $k$ is small. For instance, $k=1$
for the Strang splitting $S_2$, which is frequently used in the
literature.

Let us now consider the optimal $k$, i.e., the one minimizing
$N_{\rm new}$, for a given accuracy $\e$.
It is obtaind from the solution of the equation
$$2k^2 \ln \frac {25}3 - \ln \frac{4emt\norm{H_2}}{\e} = 0.$$
Since we seek a positive integer $k^*_{\rm new}$
minimizing $N_{\rm new}$, we
set
$$k^*_{\rm new} := 
\max\left\{ {\rm round} \left(
\sqrt{\frac 12  \log_{25/3} \frac {4emt\norm{H_2}}{\e}}\,
\,\right), 1 \right\},$$
where ${\rm round}(x) = \lfloor x + 1/2 \rfloor$, $x\ge 0$.
We can avoid using the $\max$ function in the expression above by considering
$\e\le mt\norm{H_2}$. 
Then
the number of exponentials $N_{\rm new}$ satisfies
$$
N^*_{\rm new} \leq
\frac {8}{3}\, (2m-1)\, met \, \norm{H_1} \;
e^{2
\sqrt{\tfrac 12 \ln \tfrac {25}3 \ln \tfrac {4emt\norm{H_2}}{\e}}}.
$$

Berry et al. \cite{Berry} find
\begin{equation}
k^*_{\rm prev} =  {\rm round\ }\left(
\frac 12 \sqrt{\log_5 \frac{m\norm{H_1}t}{\e}+1}
\right),
\end{equation}
which minimizes $N_{\rm pre}$. For $k^*_{\rm prev}$ the number of
exponentials $N_{\rm prev}$ becomes
\begin{equation}\label{opt_1}
N^*_{\rm prev} = 2m^2\norm{H_1}t\; e^{2
\sqrt{
\ln 5 \ln \tfrac{m\norm{H_1}t}{\e}}}.
\end{equation}

As a final comparison with $N_{\rm prev}$ we have
$$\frac {N^*_{\rm new}}{N^*_{\rm prev}} \le
 \frac {8\,e}{3}\,
e^{2
\left(
\sqrt{\tfrac 12 \ln \tfrac {25}3 \ln \tfrac {4emt\norm{H_2}}{\e}}
-  \sqrt{
\ln 5 \ln \tfrac{m\norm{H_1}t}{\e}}\, \right)
}.
$$
Hence, there is an important difference between the previously
derived optimal $k$ and the one derived in the present paper. In
\cite{Berry}, the optimal $k$ depends on $\norm{H_1}$. More precisely,
we show that the optimal $k$ depends on $\norm{H_2}$, the second
largest norm of the Hamiltonians comprising $H$, which can be considerably
smaller than $\norm{H_1}$.

\section{Acknowledgments}

We are grateful to Joseph F. Traub and Henryk Wo\'zniakowski
as well as the referees for their
insightful comments and suggestions. This research was supported in
part by National Science Foundation.

\section{Appendix}

\begin{proof}[Proof of Lemma 2]
Unwinding the recurrence for $S_{2k}$ we see that
$$
S_{2k}(\mathcal{H}_1,\dots,\mathcal{H}_m,\dt) =
\prod_{\ell=1}^K S_2((\mathcal{H}_1,\dots,\mathcal{H}_m,z_{\ell}\dt)=
\prod_{\ell=1}^K \left[ \prod_{j=1}^m e^{-i\mathcal{H}_j z_\ell\dt/2}
\prod_{j=m}^1 e^{-i\mathcal{H}_j z_\ell\dt/2}\right],$$
where $K=5^{k-1}$ and each $z_\ell$ is defined according to the recursive
scheme, $\ell=1,\dots,K$. For the details, see the part of the text
that follows (\ref{eq:largeformula}). The bound (\ref{eq:qbound0}), namely,
$$
|z_\ell| \le \frac {4k}{3^k}\quad {\rm\ for\ all\ } \ell=1,\dots,K,
$$
holds independently of $m$, because it depends on the $k-1$st levels
of the recursion tree and not on the leaf,
$S_2((\mathcal{H}_1,\dots,\mathcal{H}_m,z_{\ell}\dt)$,
at which, the corresponding to $\ell$,
path ends.

In the expression of $S_2((\mathcal{H}_1,\dots,\mathcal{H}_m,z_{\ell}\dt)$
the sum of the magnitudes of the factors multiplying the Hamiltonials
in the exponents is
$m|z_\ell|\cdot |\dt| $, for all $\ell=1,\dots,K$.
Thus in the expression of $S_{2k}$ above,
the sum of the magnitudes of all factors multiplying the Hamiltonians
in the exponents is
$$\sum_{\ell=1}^K (m|z_\ell|\cdot|\dt|) \le 5^{k-1} m \frac {4 k}{3^k}
|\dt|.$$
Define
\begin{equation}\label{eq:qbm0}
d_k:= m \frac 43 k \left(\frac 53\right)^{k-1}\quad k\ge 1.
\end{equation}

Equivalently, one can view the expression for $S_{2k}$ above as a product of
exponentials of the form $e^{\mathcal{H}_j r_{j,n}\dt}$, where
$\sum_{n=1}^{N_j}r_{j,n}=1$, $j=1,\cdots,m$, and $N_j$ is the number
of occurrences of $\mathcal{H}_j$ in $S_{2k}$.
Recall that for $m=2$
we used $s_n$ to denote $r_{1,n}$ and $z_n$ to denote $r_{2,n}$.
With this notation and using (\ref{eq:qbm0}) we have
\begin{equation}\label{eq:qbm}
\sum_{j,n}|r_{j,n}| \le d_k.
\end{equation}
(Recall the derivation of (\ref{eq:qbound}).)

Expanding the factors of $S_{2k}$ in a power series individually,
and then carrying out the multiplications amongst them, we conclude
that $S_{2k}$ is given by an infinite sum whose terms have the form
\begin{equation}\label{eq:term}
\prod_{(j,n)} \frac{1}{\gamma_{j,n}!}\,
\mathcal{H}_j^{\gamma_{j,n}}\, [-i \,r_{j,n} \,\dt]^{\gamma_{j,n}}.
\end{equation}
The factors of these products are specified by the Hamiltonians
$H_j$ and the order of their occurrences after unwinding the
recurrence for $S_{2k}$, where $j=1,\dots,m$ and
$\gamma_{j,n}=0,1,2,\dots$, for all $n=1,\dots,N_j$.

Consider the terms that contain only $\mathcal{H}_1$ and, therefore,
have $\gamma_{j,n}=0$, for $n=1,\dots,N_j$ and $j=2,\dots,m$. The
sum of these terms is
\begin{equation}
\begin{split}
\sum_{\gamma_{j,n}=0 {\rm\ for\ } j\ne 1} \prod_{(j,n)}
\frac{1}{\gamma_{j,n}!}\, \mathcal{H}_j^{\gamma_{j,n}}\, [-i
\,r_{j,n} \,\dt]^{\gamma_{j,n}} =& \sum_{\gamma_{1,1}=\cdots =
\gamma_{1,N_1}= 0}^\infty \prod_{(1,n)} \frac{1}{\gamma_{1,n}!}\,
\mathcal{H}_1^{\gamma_{1,n}}\, [-i \,r_{1,n} \,\dt]^{\gamma_{1,n}} \\
=& \prod_{n=1}^{N_1}
\sum_{\gamma_{1,n}}\frac{1}{\gamma_{1,n}!}H_1^{\gamma_{1,n}}\, [-i
r_{1,n}\,\dt]^{\gamma_{1,n}} = \prod_{n=1}^{N_1}
e^{-i\mathcal{H}_1r_{1,n}\dt} \\
& = e^{-i\sum_{n}r_{1,n}H_1\dt} = e^{-i\mathcal{H}_1\dt}.
\end{split}
\end{equation}

On the other hand,
\begin{equation}
e^{-i\sum_{j=1}^m \mathcal{H}_j\dt} = I + \left(-i\sum_{j=1}^m
\mathcal{H}_j\dt\right) + \cdots + \frac{1}{k!}\left(-i\sum_{j=1}^m
\mathcal{H}_j\dt\right)^k +\cdots,
\end{equation}
and the terms that contain only $\mathcal{H}_1$ have sum
\begin{equation}
\sum_{k=0}^{\infty} \frac{1}{k!}\mathcal{H}_1^k (-i\dt)^k =
e^{-i\mathcal{H}_1\Delta t}.
\end{equation}

Let us now consider the error bound in (\ref{suzuki2}). The sum of
the terms with only $\mathcal{H}_1$ in $S_{2k+1}$ and
$\exp(\sum_{j=1}^m H_j\dt)$ is the same and cancels out when we
subtract one from the other. Moreover, in $\exp(-i\sum_{j=1}^m
\mathcal{H}_j\dt)-S_{2k}(\dt)$ we know that the terms of order up to
$2k$ also cancel out, see Eq. (\ref{suzuki2}). From this we conclude
that the error is proportional to
$\norm{\mathcal{H}_2}|\dt|^{2k+1}$.

Consider
\begin{equation}
\exp(-i(\mathcal{H}_1+\cdots+\mathcal{H}_m)\dt)-
S_{2k}(\mathcal{H}_1,\dots,\mathcal{H}_m,\dt) =
\sum_{l=2k+1}^{\infty} \bigl[ R_{l}(\dt)-T_{l}(\dt)\bigr],
\end{equation}
where $R_{l}(\dt)$ is the sum of all terms in
$\exp(-i(\mathcal{H}_1+\cdots+ \mathcal{H}_m)\dt)$ corresponding to
$\dt^{l}$ and $T_{l}(\dt)$ is the sum of all terms in $S_{2k}$
corresponding to $\dt^l$. We can ignore the terms in $T_{l}(\dt)$
and $R_{l}(\dt)$ that contain only $\mathcal{H}_1$ (and not
$\mathcal{H}_2$) as a factor.

Then
\begin{equation}
\norm{R_{l}(\dt)} = \left\| \frac{1}{l!}\left(\sum_{j=1}^m
\mathcal{H}_j\dt\right)^l - \frac{1}{l!}\mathcal{H}_1^l\dt^l\right\|
\leq \frac{m^{l}}{l!}\norm{\mathcal{H}_2}|\dt|^{l},
\end{equation}
because there are $m^{l}-1$ terms in $R_l$ and each norm is at most
$\frac{1}{l!}\norm{\mathcal{H}_2}|\dt|^{l}$.

From (\ref{eq:term}) we have
\begin{equation}
T_{l}(\dt) = \sum_{\sum \gamma_{(j,n)}=l}
\frac{\prod_{(j,n)}r_{j,n}^{\gamma_j,n}}{\prod_{(j,n)}\gamma_{j,n}!}
\prod_{(j,n)}\mathcal{H}_j^{\gamma_{j,n}}\dt^{l},
\end{equation}
where $\sum_{n}\gamma_{1,n} \neq l$, i.e., there is no terms
containing only $\mathcal{H}_1$. So,
$\norm{\prod_{(j,n)}\mathcal{H}_j^{\gamma_{j,n}}}\leq
\norm{\mathcal{H}_2}$, and
\begin{equation}\label{1}
\norm{T_{l}(\dt)} \leq \sum_{\sum \gamma_{j,n}=l}
\frac{\prod_{j,n}|r_{j,n}|^{\gamma_j,n}}{\prod_{j,n}\gamma_{j,n}!}\norm{\mathcal{H}_2}
|\dt|^{l}.
\end{equation}
To calculate the coefficients of the sum, we consider
\begin{equation}\label{2}
\begin{split}
\prod_{(j,n)} \exp(|r_{j,n}\dt|) &=
\prod_{(j,n)}\sum_{\gamma_{j,n}=0}^{\infty}
\frac{1}{\gamma_{j,n}!} |r_{j,n}\dt|^{\gamma_{j,n}}\\
&=\sum_{l=0}^{\infty} \sum_{\sum \gamma_{j,n}=l}
\frac{\prod_{j,n}|r_{j,n}|^{\gamma_j,n}}{\prod_{j,n}\gamma_{j,n}!}
|\dt|^{l}.
\end{split}
\end{equation}
Hence the coefficient of $|\dt|^l$ in (\ref{1}) is equal to that in
(\ref{2}). Also
\begin{equation}
\prod_{j,n} \exp(|r_{j,n}\dt|) = \exp(\sum_{j,n}|r_{j,n}\dt|).
% next two lines from old text - ignore here
%=
%\exp(m\dt) = \sum_{l=0}^{\infty} \frac{m^l}{l!} \dt^l,
\end{equation}

From (\ref{eq:qbm}) we obtain
\begin{equation}
\norm{T_{l}(\dt)} = \frac{d_k^l}{l!}\norm{\mathcal{H}_2}|\dt|^{l}.
\end{equation}
Therefore,
\begin{equation}
\begin{split}
\norm{\exp(\sum_{j=1}^m \mathcal{H}_j\dt) - S_{2k}(\dt)}&\leq \sum_{l=2k+1}^{\infty}\norm{R_l(\dt)}+\norm{T_l(\dt)}\\
\leq
&2 \sum_{l=2k+1}^{\infty}\frac{d_k^l}{l!}
\norm{\mathcal{H}_2} |\dt|^{l}\\
=&2\norm{\mathcal{H}_2}\sum_{l=2k+1}^{\infty}\frac{1}{l!}|d_k\dt|^{l} \\
&\leq \frac{2}{(2k+1)!}\norm{\mathcal{H}_2}|d_k\dt|^{2k+1}
\left( 1 -\frac {d_k|\dt|}{2k+2}\right)^{-1} \\
&\le
\frac{4}{(2k+1)!}\norm{\mathcal{H}_2} |d_k\dt|^{2k+1},
\end{split}
\end{equation}
where the last two inequalities follow from
the assumption $d_k|\dt|\le k+1$ and
an estimate of the tail of
the Poisson distribution; see, e.g., \cite[Thm 1]{possion}.
\end{proof}

\end{document}